\documentclass[11pt]{article}
\pdfoutput=1
\usepackage[totalwidth=480pt, totalheight=680pt]{geometry}

\usepackage[latin1]{inputenc}

\usepackage{hyperref}

\usepackage{amsmath}
\usepackage{amssymb}
\usepackage{amsthm}

\usepackage{latexsym}
\usepackage{bbm}

\usepackage{booktabs}

\usepackage[above,below]{placeins}
\usepackage[title,titletoc]{appendix}

\usepackage{ifpdf}

\ifpdf
\usepackage[pdftex]{graphicx}
\else
\usepackage{graphicx}
\fi

\graphicspath{{curves/}}

\begin{document}

\ifpdf
\DeclareGraphicsExtensions{.pdf, .jpg, .tif}
\else
\DeclareGraphicsExtensions{.eps, .jpg}
\fi

\newtheorem{prop}{Proposition}

\theoremstyle{remark}
\newtheorem{rem}{Remark}

\thispagestyle{empty}

%\maketitle

\vspace*{3cm}

\begin{center}
\begin{Large}
\textbf{One-Dimensional Pricing of CPPI}
\end{Large}

\vspace{10mm} {\bf Louis Paulot and Xavier Lacroze}

\vspace{7mm}
\emph{Sophis Technology
\\
24--26 place de la Madeleine, 75008 Paris, France}

\vspace{5mm}
{\ttfamily louis.paulot@sophis.net\\xavier.lacroze@sophis.net}

\vspace{2cm}
May 18, 2009

\emph{Revised} February 8, 2010

\end{center}

\vspace{2cm}
\hrule
\begin{abstract}

Constant Proportion Portfolio Insurance (CPPI) is an investment strategy designed to give participation in the performance of a risky asset while protecting the invested capital. This protection is however not perfect and the gap risk must be quantified. CPPI strategies are path-dependent and may have American exercise which makes their valuation complex. A naive description of the state of the portfolio would involve three or even four variables. In this paper we prove that the system can be described as a discrete-time Markov process in one single variable if the underlying asset follows a process with independent increments. This yields an efficient pricing scheme using transition probabilities. Our framework is flexible enough to handle most features of traded CPPIs including profit lock-in and other kinds of strategies with discrete-time reallocation.

\end{abstract}
\hrule

\vspace{\stretch{1}}

\pagebreak

\setcounter{tocdepth}{3}
\tableofcontents

%\pagebreak

\parskip=4pt

\section{Introduction}

Constant Proportion Portfolio Insurance is a dynamic strategy designed to give participation in risky assets while protecting the invested capital \cite{merton1970oca}. This is achieved by periodically rebalancing between a risk-free asset (Zero-Coupon bond) and a risky asset (share, index, fund, fund of funds\ldots). In the simplest form, if the underlying asset has no jumps and if one can rebalance continuously, the final payoff depends in a deterministic way on the risky underlying. However both hypotheses are strong and do not fit real market conditions. If they are relaxed, the strategy is not as efficient: there is a small chance of not recovering at maturity the invested capital. This gap risk may be held by the issuer, so that the principal is really guaranteed to the investor. In this case, there is an option included in the product, which must be priced and hedged. For a very simple CPPI strategy with continuous rebalancing, the gap risk comes only from instantaneous jumps and can be quantified analytically \cite{cont:cpp}. With a discrete rebalancing scheme, there is a closed formula for the embedded option price if the underlying follows a Black-Scholes diffusion \cite{balder:ecs}. This formula can as well be generalized to jump-diffusion models, and more generally to Levy processes \cite{paulot2009epc}.  However these methods work only for an idealized CPPI product where there are no caps on the risky exposure, no spreads on the risk-free and financing rates, no fees, no profit lock-in, a natural bond floor... Unfortunately, real CPPIs have usually such features, which prevent from using closed formulas. Moreover, such formulas does not hold for options with strikes differing from the guaranteed amount. This makes necessary to use other methods for real-life, more complex products. As CPPIs are very path dependent, they are usually priced through Monte-Carlo simulations (see \cite{boulier} for an example). Extreme value theory has also been used to estimate the gap risk of such products \cite{bertrand2002pie}.

Monte-Carlo pricing is perfectly suited to path dependence but the dimension of the problem is generally quite large. As an example, a monthly CPPI defined on a single risky asset with monthly rebalancing and a 10 years maturity requires simulating 120 values of the underlying. This means a 120-dimensional Monte-Carlo integration. Furthermore, the tails of the distributions are crucial in the pricing. First of all, the lower tail gives the gap risk part. It must be computed with enough precision to produce a reasonable Put price and an accurate Delta. The upper tail is important to reach the correct mean value of the CPPI strategy. The payoff distribution is close to a shifted lognormal distribution with high leverage: most trajectories will end below the mean and will be compensated by a few very high terminal values.
%The control variate can solve part of the convergence problems for the strategy mean value, by comparing the payoff of the real CPPI to the payoff of a self-financing CPPI. However that does not solve all convergence issues: in particular the optional part, the gap risk, is not smoothed by basic control variate.
The high path dependence and barrier-like structure of the strategy makes the Put price and its Delta converge very slowly. As an illustration, we refer to figure~\ref{convergence_mc} in which we present an example of the convergence of a put option on the CPPI portfolio. One needs 100~millions Monte-Carlo simulations (and hours of computation) to get a rough estimation of the price of the derivative.
\begin{figure}[b]
\centering
\includegraphics[width=\textwidth]{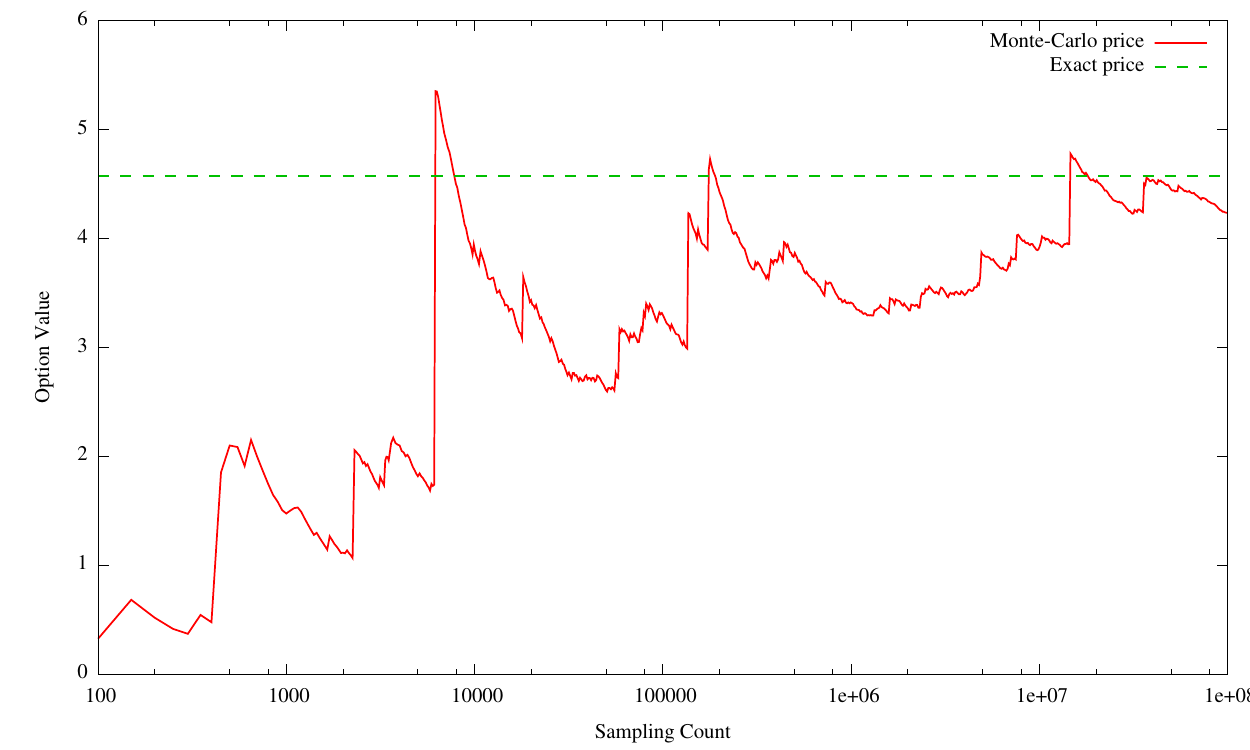}
\caption{\emph{Convergence of the Monte Carlo method for the price of a put on a CPPI portfolio as functions of the number of paths. Maturity is 10 years, multiplier is 4, reallocation is performed monthly, volatility is 35\% and initial investment is 1\,000. The strike is set to this value. The dashed line indicates the best estimate obtained in section \ref{sec:numres}.}}
\label{convergence_mc}
\end{figure}

An other way to price CPPI derivatives would be to use a PIDE scheme. This would also handle the American or more precisely Bermudan exercise which is more and more present in such products. However, there are at least three variables to propagate which depend on each other: the risky underlying value, the CPPI portfolio value and the risk asset weighting. If some profit lock-in feature is present, the guarantee level increases the dimension of the PIDE to four. The high dimensionality combined with the large number of time steps needed (at least the number of reallocations) makes this solution very computationally expensive.

There are drastic simplifications if one considers the portfolio value at rebalancing times only. We prove that if the underlying asset has independent increments, the dynamics of the portfolio at rebalancing dates is described by a discrete time Markov process in one single variable. This property leads to a simple and efficient numerical scheme. The main advantages of this technique are its speed, the smoothness of greeks and the ability to handle Bermudan exercise. For example in the simplest cases, only a few milliseconds are needed to get an accurate price. More complex cases with profit lock-in, artificial cushion or coupons take more time but never more than a few seconds.

In section \ref{sec:basic}, we start by describing our pricing method in simple cases for which the guarantee amount is constant: after introducing our working hypothesis, we formulate the main result concerning the transition probabilities of the CPPI portfolio at rebalancing dates. Then we discuss the numerical implementation of this method in section \ref{sec:numericalBasic}. In section \ref{sec:lockin}, we allow the guaranteed amount to vary over time through a feature called "profit lock-in". Under some homogeneity assumption on the rebalancing rule, we show that the problem is still one-dimensional and discuss the steps involved in the numerical implementation.

\section{Basic algorithm}
\label{sec:basic}

\subsection{Sketch of the method}
\label{sec:sketch}

At a time $t$, we consider a portfolio (the CPPI portfolio) with total value $C(t)$. This portfolio is split between a risky and a risk-free components:
\[
    C(t) = C^\text{risky}(t) + C^\text{risk-free}(t)
\]
The risk-free part is invested in zero-coupons of same maturity as the CPPI, $B(t)$, whose dynamic is completely deterministic
\[
    B(t') = B(t) e^{r(t'-t)}
\]
with $r$ the risk-free interest rate (taken constant for simplicity). The risky part is invested in an asset $S(t)$. For simplicity in this part we consider a Black-Scholes dynamics which under the risk-neutral measure is
\[
    \mathrm d S(t) = r S(t) \mathrm d t + \sigma S(t) \mathrm d W(t)
\]

A CPPI strategy is designed to protect some part of the investment and we shall denote by $G$ this "guaranteed" amount. At this point, $G$ is just some constant, which is generally taken to be the initial value of the portfolio $C(0)$. Taking a fraction $1/G$ of the portfolio, we can consider without loss of generality a unit guarantee. In this section we take therefore $G=1$.

Let $t_0 = 0 < t_1 < \dots < t_n$ be some fixed times, called rebalancing dates. At each rebalancing date~$t_i$, the allotment of the portfolio between the risky and risk-free components is adjusted (without any change in the total value of the portfolio) according to a rule specified through a non-negative deterministic function $w\!\left(t, x\right)$ called the risky asset weighting (or simply RAW):
\begin{equation}
\label{eq:rebalancing_rule}
	C^\text{risky}(t_i) = w \big(t_i, C(t_i) \big) C(t_i)
\end{equation}
For a CPPI strategy, $w$ is usually chosen of the form
\begin{equation}
\label{eq:rebalancing_rule2}
    w \big(t, C \big) = \max\!\left( m \frac{C-H(t)}{C}, 0 \right)
\end{equation}
where $m$ is called the multiplier of the strategy and $H$ is a positive deterministic function called the floor whose terminal value is the guarantee (equal to 1 here). A natural choice is to set the floor to be
\[
    H(t) = e^{-r(t_n-t)}
\]
$H(t)$ is equal to the value of a (unit nominal) zero-coupon bond one must hold to recover 1 at maturity $t_n$. Within this setup, this investment strategy indeed offers some protection: when the portfolio poorly performs the investment is smoothly allocated to the risk-free asset. With continuous rebalancing and without jumps in the asset price, the bond floor cannot be breached. With discrete rebalancing there is still a probability of going under the floor: the floor is breached if the risky asset goes down by more than $1/m$ between two rebalancing dates.

Between rebalancing dates, the composition of the portfolio is left unchanged. For $t \in [t_i,t_{i+1}]$, this reads
\begin{equation}
	\label{eq:portfolio_dynamic}
	C(t) = C^\text{risky}(t_i) \frac{S(t)}{S(t_i)} + C^\text{risk-free}(t_i) \frac{B(t)}{B(t_i)}
\end{equation}
In the case $t=t_{i+1}$, this equation gives the evolution of the portfolio between two consecutive rebalancing dates:
\begin{equation*}
    C(t_{i+1}) = C^\text{risky}(t_{i}) \frac{S(t_{i+1})}{S(t_i)} + C^\text{risk-free}(t_{i}) \frac{B(t_{i+1})}{B(t_i)}
\end{equation*}
or equivalently, eliminating $C^\text{risky}$ using equation \eqref{eq:rebalancing_rule},
\begin{equation}
\label{eq:portfolio}
    C(t_{i+1}) =
    w \big(t_i, C(t_i) \big) C(t_{i}) \frac{S(t_{i+1})}{S(t_i)} + \left(1 - w \big(t_i, C(t_i) \big) \right) C(t_{i}) \frac{B(t_{i+1})}{B(t_i)}
\end{equation}
Then, at date $t_{i+1}$, the manager computes the new Risky Asset Weighting $w \big(t_{i+1}, C(t_{i+1}) \big)$ and reallocates the resources. Once again, we emphasize that this rebalancing procedure has no instantaneous effect on the total value of the portfolio.

In a first step, we look at the risk-neutral probability distribution of the portfolio value over a single rebalancing period. It can be seen from equation~\eqref{eq:portfolio} that $C(t_{i+1})$ depends on $C(t_i)$ and on the ratio of the risky asset values $S(t_{i+1}) / S(t_{i})$. However, if the process $\log S(t)$ has independent increments (which clearly is the case of the lognormal dynamic we consider here), this ratio does not depend on $S(t_i)$ itself. If in addition we suppose that the risky asset has independent increments, the distribution of $C(t_{i+1})$ conditionally on time $t_i$ depends only on $C(t_i)$:
\begin{eqnarray}
\mathbb P \Big[ C(t_{i+1}) < y \, \big\vert \, \mathcal{F}_{t_i} \Big] & = &
\mathbb P \Big[ C(t_{i+1}) < y \, \big\vert \, C(t_i) \Big] \nonumber \\
   & = & \mathbb P \left[ \frac{S(t_{i+1})}{S(t_i)} < \frac{y - \left(1 - w \big(t_i, C(t_i) \big) \right) C(t_i) B(t_{i+1}) / B(t_i) }{w \big(t_i, C(t_i) \big) C(t_{i})} \, \bigg\vert \, C(t_i) \right]
\label{eq:singleperiod1}
\end{eqnarray}
This equation relates the distribution of $C(t_{i+1})$ conditionally on $C(t_i)$ to the distribution of $S(t_{i+1}) / S(t_i)$. In the case of Black-Scholes diffusion that we consider, the cumulative function is
\begin{equation}
\mathbb P \left[ \frac{S(t_{i+1})}{S(t_i)} < z \right] = \mathcal N\!\left( \frac{\ln z - \left(r - \frac{1}{2} \sigma^2 \right)(t_{i+1} - t_i)}{ \sigma \sqrt{t_{i+1}-t_i}}\right)
\label{eq:singleperiod2}
\end{equation}
Combining these two equations we get an explicit formula for $\mathbb P \Big[ C(t_{i+1}) < y \, \big\vert \, \mathcal{F}_{t_i} \Big]$ as a function of $C(t_i)$. The process for the portfolio value at rebalancing times is therefore a Markov process in discrete time. The general framework for the study of Markov processes leads to introducing the risk-neutral probability density of $C(t_j)$ conditionally on $C(t_i)$, also called transition kernel for the period $[t_i,t_{i+1}]$:
\begin{equation*}
    \Phi_{t_i,t_{i+1}} (x,y) = \frac{\partial}{\partial y} \mathbb P \Big[ C(t_{i+1}) < y \, \big\vert \, C(t_i) = x \Big]
%\label{eq:cumudensity}
\end{equation*}
which is computed explicitly using equations~\eqref{eq:singleperiod1} and~\eqref{eq:singleperiod1}.

Let us consider a European derivative product $V$ written on the CPPI portfolio with terminal payoff at $t_n$ given by
\[
    V\!\big(t_n,C(t_n)\big) = P\big(C(t_n)\big)
\]
If we know the value of this derivative at time $t_{i+1}$ for all values of $C(t_{i+1})$, its value at time $t_i$ and CPPI value $C(t_i)$  is given by the discounted risk-neutral expected value
\begin{eqnarray*}
V\!\big(t_i, C(t_i)\big) & = & e^{-r(t_{i+1} - t_i)} \ \mathbb E\Big[ V\!\big(t_{i+1}, C(t_{i+1})\big) \, \big\vert \, C(t_i) \Big] \nonumber\\
& = & e^{-r(t_{i+1} - t_i)} \int \mathrm d y \ \Phi_{t_i,t_{i+1}}\!\big(C(t_{i}),y\big) \ V\!\big( t_{i+1}, y\big)
\end{eqnarray*}
Starting from the terminal payoff and applying recursively this formula yields today fair price of the derivative $V\!\big(t_0, C(t_0)\big)$.
%as
%\begin{multline}
%\label{eq:derivativeprice}
%V(t_0, C(t_0)) = \int \mathrm d C(t_{1}) \ \Phi_{t_0,t_{1}}\!\big(C(t_{0}),C(t_{1})\big) \int \mathrm d C(t_{2}) \ %\Phi_{t_1,t_2}\!\big(C(t_1),C(t_2)\big) \ \cdots
%\\
%\cdots \int \mathrm d C(t_{n}) \ \Phi_{t_{n-1},t_{n}}\!\big(C(t_{n-1}),C(t_{n})\big) \ V\!\big( t_{n}, C(t_{n})\big)
%\end{multline}

This backward recursion is also suited to Bermudan exercise (on rebalancing dates):
\begin{equation*}
V\!\big(t_i, C(t_i)\big) = \max\!\left( e^{-r(t_{i+1} - t_i)} \int \mathrm d y \ \Phi_{t_i,t_{i+1}}\!\big(C(t_{i}),y\big) \ V\!\big( t_{i+1}, y\big) \ , \ P\big( C(t_i)\big) \right)
\end{equation*}

\subsection{Mathematical formulation}

We collect the previous results (with slightly more general hypothesis) in the following proposition.
\begin{prop}
\label{prop:markov}
Let us consider a CPPI portfolio with a given set of rebalancing dates $t_0,\dots,t_n$ and guaranteed amount $G$. If the following hypotheses hold:
\begin{itemize}
\item the risky asset weighting is given by a deterministic function $w(t,C/G)$;
\item the underlying risky asset follows a one-dimensional process, its logarithm has independent increments and it has a density
\begin{equation*}
 \mathcal \phi_{t,t'}(z) = \partial_z \mathbb P \left[ \frac{S(t')}{S(t)} < z \right]
\end{equation*}
 which depends only on $t$, $t'$ and $z$;
\item the risk free asset is deterministic with value $B(t)$;
\end{itemize}
then the CPPI portfolio value taken at rebalancing dates $C(t_i)$ is a Markov process in discrete time. The probability transition operator of the normalized variable $C(t_i)/G$
\begin{equation*}
\Phi_{t_i,t_{i+1}} (x,y) = \mathbb E\!\left[ \delta\!\left(\frac{C(t_{i+1})}{G}-y\right) \, \Big\vert \, \frac{C(t_i)}{G} = x \right]
\end{equation*}
given by
\begin{equation}
\Phi_{t_i,t_{i+1}} (x,y) = \left\{
\begin{array}{ll}
    \displaystyle \frac{1}{w \big(t_i, x \big) x} \ \phi_{t_i,t_{i+1}}\!\left[\frac{y - \left(1 - w \big(t_i, x \big) \right) x B(t_{i+1}) / B(t_i) }{w \big(t_i, x \big) x} \right] & \mathrm{\ if\ }w \big(t_i, x \big) x > 0
        \vspace{6pt}
    \\
    \displaystyle \delta\!\left(y - x \frac{B(t_{i+1})}{B(t_i)} \right) & \mathrm{\ if\ }w \big(t_i, x \big) x = 0
\end{array} \right.
\label{eq:cumudensity}
\end{equation}
European derivatives on the CPPI portfolio are valued recursively through
\begin{equation}
V\!\big(t_i, C(t_i)\big) = e^{-r(t_{i+1} - t_i)} \int \mathrm d y \ \Phi_{t_i,t_{i+1}}\!\left(\frac{C(t_{i})}{G},y\right) \ V\!\big( t_{i+1}, G\, y\big)
\label{eq:recursion}
\end{equation}
and Bermudan derivatives using
\begin{equation}
\label{eq:recursionAmerican}
V\!\big(t_i, C(t_i)\big) = \max\!\left( e^{-r(t_{i+1} - t_i)} \int \mathrm d y \ \Phi_{t_i,t_{i+1}}\!\left(\frac{C(t_{i})}{G},y\right) \ V\!\big( t_{i+1}, G\, y\big) \ , \ P\big( C(t_i)\big) \right)
\end{equation}
\end{prop}

\begin{proof}
The proof goes along the lines of the last section for a fraction $1/G$ of the portfolio. The explicit expressions for the risk-free asset and the floor are simply relaxed and the lognormal cumulative law is replaced by a more general cumulative law. The Markovian property stems from the fact that the transition probability depends only on $C(t_i)/G$.
\end{proof}

\begin{rem}
The condition on the density between rebalancing dates excludes stochastic volatility and more generally models with hidden variables. In fact such models can be used within our framework at the cost of increasing the dimensionality of the Markov process by the number of additional variables. For example a stochastic volatility model for the underlying will be converted into a 2-dimensional discrete-time Markov process. As matrix-vector multiplications on a grid of size $N$ have a complexity $O(N^2)$, adding a second dimension with grid size $M$ will increase by a factor $M^2$ the computation time. Therefore it looks more reasonable to use a volatility regime-switching model rather than a model with continuous stochastic volatility. In such a model, there are a small number $M$ of regimes, each with its own volatility. One computes the joint probability of the CPPI portfolio value and the volatility regime conditionally on the initial portfolio value and volatility regime.
\end{rem}

\section{Numerical implementation}
\label{sec:numericalBasic}

\subsection{Order 2 scheme}

The numerical implementation is rather direct from proposition \ref{prop:markov}.
\begin{itemize}
\item
The normalized value $C(t_i)/G$ of the portfolio was previously allowed to be any real number. This real line is discretized into a grid $h_j$, $j=0\ldots N$. The (rescaled) initial value of the portfolio $C(t_0)/G$ is supposed to lay on the $j_0^\text{th}$ grid node: $C(t_0)=G\, h_{j_0}$.
\item
The transition operator of equation \eqref{eq:cumudensity} acting on real functions is replaced with a transition matrix  $M^{(i,i+1)}_{jk}$ acting on vectors.\\ 
Let $\mathcal Q_{t_i,t_{i+1}}$ and $\mathcal Q_{t_i,t_{i+1}}^{(1)}$ denote the cumulative density and partial expected value of the underlying process
\begin{equation*}
\begin{array}{rclcl}
\mathcal Q_{t_i,t_{i+1}}(z) &=& \displaystyle\mathbb P \left[ \frac{S(t')}{S(t)} < z \right]
&=& \displaystyle \int_{-\infty}^z \!\mathrm d y \, \phi_{t_i,t_{i+1}}(y)
\\
\mathcal Q_{t_i,t_{i+1}}^1(z) &=& \displaystyle\mathbb E \left[ \frac{S(t')}{S(t)} \mathbbm{1}_{]-\infty ; z[}\!\left(\frac{S(t')}{S(t)}\right) \right]
&=& \displaystyle \int_{-\infty}^z \!\mathrm d y  \, \phi_{t_i,t_{i+1}}(y) \, y
\end{array}
\end{equation*}

Let $X^{(i)}_{jk}$ denote the value of the underlying return $\frac{S(t_{i+1})}{S(t_i)}$ for which a portfolio worth $G h_j$ at $t_i$ is worth $G h_k$ at $t_{i+1}$: 
\begin{equation*}
X^{(i)}_{jk} = \frac{h_k - \left(1 - w \big(t_i, h_j \big) \right) h_j B(t_{i+1}) / B(t_i) }{w \big(t_i, h_j \big) h_j}
\end{equation*}

Starting from a grid point $h_j$ at time $t_i$,
for $w \big(t_i, h_j \big) h_j > 0$ the probability of being between $h_k$ and $h_{k+1}$ at time $t_{i+1}$ is the finite difference
\begin{equation*}
Q^{(i,i+1)}_{jk} = \displaystyle\mathbb P \left[ h_k \leq \frac{C(t_{i+1})}{G} < h_{k+1} \ \Big\vert \   \frac{C(t_i)}{G} = h_j \right]
 =  
\mathcal Q_{t_i,t_{i+1}}\!\left(X^{(i)}_{j,k+1} \right) - \mathcal Q_{t_i,t_{i+1}}\!\left(X^{(i)}_{jk} \right)
\end{equation*}
Similarly the partial expected value is
\begin{eqnarray*}
Q^{1 \, (i,i+1)}_{jk} &=& \displaystyle\mathbb E\!\left[ \frac{C(t_{i+1})}{G} \mathbbm{1}_{ [h_k ; h_{k+1} [}\!\left( \frac{C(t_{i+1})}{G} \right) \ \Big\vert \   \frac{C(t_i)}{G} = h_j \right]
\\
 &=&  w \big(t_i, h_j \big) h_j 
\bigg[ \mathcal Q^1_{t_i,t_{i+1}}\!\left(X^{(i)}_{j,k+1}\right) - \mathcal Q^1_{t_i,t_{i+1}}\!\left(X^{(i)}_{jk}\right) \bigg]
\\
&& + \left(1 - w \big(t_i, h_j \big) \right) h_j \frac{B(t_{i+1})}{B(t_i)}  \bigg[ \mathcal Q_{t_i,t_{i+1}}\!\left(X^{(i)}_{j,k+1}\right) - \mathcal Q_{t_i,t_{i+1}}\!\left(X^{(i)}_{jk}\right) \bigg]
\end{eqnarray*}

For $w \big(t_i, h_j \big) h_j = 0$ we have respectively
\begin{eqnarray*}
Q^{(i,i+1)}_{jk} &=& \mathbbm 1_{\left[h_k;h_{k+1}\right[}\!\left(h_j\frac{B(t_{i+1})}{B(t_i)}\right)
\\
Q^{1\, (i,i+1)}_{jk} &=& h_j\frac{B(t_{i+1})}{B(t_i)} \, \mathbbm 1_{\left[h_k;h_{k+1}\right[}\!\left(h_j\frac{B(t_{i+1})}{B(t_i)}\right)
\end{eqnarray*}

Using these matrices, we take as transition matrix between $t_i$ and $t_{i+1}$
\begin{equation*}
M^{(i,i+1)}_{jk} = M^{+ \, (i,i+1)}_{jk} + M^{- \, (i,i+1)}_{jk}
\end{equation*}
with
\begin{eqnarray*}
M^{+ \, (i,i+1)}_{jk} &=& \frac{h_{k+1} Q^{(i,i+1)}_{jk} - Q^{1\, (i,i+1)}_{jk}}{h_{k+1}-h_k}
\\
M^{- \, (i,i+1)}_{jk} &=& \frac{Q^{1\, (i,i+1)}_{j,k-1} - h_{k-1} Q^{(i,i+1)}_{j,k-1}}{h_k-h_{k-1}}
\end{eqnarray*}
These matrices are such that the probability  $Q^{(i,i+1)}_{jk}$ to be inside the interval $[h_k,h_{k+1}[$ and the partial mean $Q^{1\, (i,i+1)}_{jk}$ are exact:
\begin{eqnarray*}
\int_{h_k}^{h_{k+1}}\!\mathrm d y \, \Phi_{(t_i,t_{i+1})}(h_j,y)\quad = & Q^{(i,i+1)}_{jk} & =\quad M^{+ \, (i,i+1)}_{jk} + M^{- \, (i,i+1)}_{j,k+1}
\\
\int_{h_k}^{h_{k+1}}\!\mathrm d y \, \Phi_{(t_i,t_{i+1})}(h_j,y)\,y \quad = & Q^{1\, (i,i+1)}_{jk} & =\quad M^{+ \, (i,i+1)}_{jk} h_k + M^{- \, (i,i+1)}_{j,k+1} h_{k+1}
\end{eqnarray*}

\item For a European product, we define a vector $V^{(i)}_j$ with the value of the derivative at date $t_i$ conditionally to $C(t_i)=G\, h_j$. The backward propagation \eqref{eq:recursion} is turned into the matrix--vector product
\begin{equation*}
    V^{(i)}_j = e^{-r(t_{i+1} - t_i)} \sum_k M^{(i,i+1)}_{jk} \, V^{(i+1)}_k
\end{equation*}
with terminal condition
\begin{equation*}
    V^{(n)}_j = P(G\, h_j)
\end{equation*}
The fair price of the derivative at time 0 is given by
\begin{equation*}
    V(t_0,C(t_0)) = V^{(0)}_{j_0}
\end{equation*}
\item In the case of a Bermudan product the recursion is
\begin{equation*}
    V^{(i)}_j = \max\!\left( e^{-r(t_{i+1} - t_i)} \sum_k M^{(i,i+1)}_{jk} \, V^{(i+1)}_k \ ,\ P(G\, h_j) \right)
\end{equation*}
\end{itemize}

\subsection{Error analysis}
\label{sec:error}

We now sketch an analysis of the numerical error for a European payoff.
Denoting by $N$ the number of grid points and omitting the discount factor, the error in backward propagation on one time step from $t_{i+1}$ to $t_i$  is
\begin{eqnarray*}
e_i(h_j) &=& \sum_{k=0}^N M^{(i,i+1)}_{jk} V^{(i+1)}(h_k) - \int \mathrm d y \, \Phi_{t_i,t_{i+1}} (h_j,y) V^{(i+1)}(y) 
\\
&=& \sum_{k=0}^{N-1} \left[  M^{+ \, (i,i+1)}_{jk} V^{(i+1)}(h_k) +  M^{- \, (i,i+1)}_{j,k+1} V^{(i+1)}(h_{k+1}) - \int_{h_k}^{h_{k+1}}\!\mathrm d y \, \Phi_{t_i,t_{i+1}} (h_j,y) V^{(i+1)}(y) 
  \right] 
\end{eqnarray*}
$M^{+ \, (i,i+1)}_{jk}$ and $M^{- \, (i,i+1)}_{j,k+1}$ have been defined such that they give exact results on constant and linear terms on each interval $[h_k;h_{k+1}[$:
\begin{multline*}
M^{+ \, (i,i+1)}_{jk} V(h_k) +  M^{- \, (i,i+1)}_{j,k+1} V(h_{k+1}) = 
\\
\int_{h_k}^{h_{k+1}}\!\mathrm d y \, \Phi_{t_i,t_{i+1}}(h_j,y) \left[ \frac{h_{k+1}-y}{h_{k+1}-h_k} V^{(i+1)}(h_k) + \frac{y-h_k}{h_{k+1}-h_k} V^{(i+1)}(h_{k+1}) \right]
\end{multline*}
As in the case of a trapezoidal quadrature, we are left with a second order term
\begin{equation*}
e_i(h_j) = \sum_{k=0}^{N-1}\left[ \frac{1}{12} \Phi_{t_i,t_{i+1}}(h_j,h_k) {V^{(i+1)}}''\!(h_k)\,  \Delta h_k^3 + o\!\left(\Delta h_k^3\right) \right]
\end{equation*}
with $\Delta h_k = h_{k+1}-h_k$. The second derivative can be evaluated by finite differences at the same order of precision. We denote it by $D^2 V_k^{(i)}$.

The error computed at time $t_i$ and grid node $h_j$ must then be backward propagated to time $t_0$ to get the contribution to the final price. Summing errors from all time steps and grid points we get the total error for a European payoff
\begin{equation*}
e = \frac{1}{12} \sum_{k=0}^{N-1} \left[ \left( \sum_{i=1}^n e^{-r (t_i- t_0)} M^{(0,i)}_{j_0,k}\, D^2 V_k^{(i)} \right) \Delta h_k^3 + o\!\left(\Delta h_k^3\right) \right]
\end{equation*}
As $\Delta h_k$ is proportional to $N^{-1}$, we end up with an asymptotic error
\begin{equation*}
e(h_j) = O\!\left(N^{-2}\right)
\end{equation*}
Moreover, minimizing the error under the constraint of a constant $\displaystyle \sum_k \Delta h_k$ gives an optimal grid step
\begin{equation}
\label{eq:optimalgrid1}
\Delta h_k \propto \frac{1}{\sqrt{ \displaystyle \left\vert \sum_{i=1}^n e^{-r (t_i- t_0)} M^{(0,i)}_{j_0,k}\, D^2 V_k^{(i)}\right\vert}}
\end{equation}
with $M^{(0,i)}$ the transition matrix between time $t_0$ and $t_i$: $M^{(0,i)} = \displaystyle \prod_{l=0}^{i-1} M^{(l,l+1)}$. In order to estimate the optimal grid, one can either use an approximative analytical solution for the portfolio values and transition matrices or run a first round of the algorithm using a basic coarser grid as an initial guess.

\subsection{Higher order schemes}

The algorithm can be generalized to reach a  $O(N^{-p})$ convergence for any $p$. To construct the transition matrix in the second order case, we have weighted the boundaries of all grid intervals in order to match both the probability to be in the interval and the partial mean on the interval. Introducing $p-2$ intermediate grid points inside all intervals allows to match all conditional moments on every interval up to order $p-1$. Following the lines of the error analysis of section \ref{sec:error}, one gets an $O(N^{-p})$ error
\begin{equation*}
e \propto \sum_{k=0}^{N-1} \left[ \left( \sum_{i=1}^n e^{-r (t_i- t_0)} M^{(0,i)}_{j_0,k}\, D^p V_k^{(i)} \right) \Delta h_k^{p+1} + o\!\left(\Delta h_k^{p+1}\right) \right]
\end{equation*}
with an optimal grid for a European payoff given by
\begin{equation}
\label{eq:optimalgrid2}
\Delta h_k \propto \frac{1}{\sqrt[p]{\displaystyle \left\vert \sum_{i=1}^n e^{-r (t_i- t_0)} M^{(0,i)}_{j_0,k}\, D^p V_k^{(i)}\right\vert }}
\end{equation}
There is however a drawback: this (partial) moment matching procedure may introduce negative transition probabilities on some grid points, which could result in numerical instabilities.

\subsection{Numerical results}
\label{sec:numres}

We consider a put on a CPPI portfolio with strike at the guaranteed amount. Maturity is ten years, the portfolio is reallocated monthly with multiplier 4; the bond floor is linear starting at 75\% so that the initial investment in the risky asset is 100\%. The initial portfolio value is 1000. On rebalancing dates, proportional fees are removed from the portfolio value at an annual rate of 30~bp. For our test we choose a constant, flat 35\% volatility. The yield curve is the EUR yield curve on January 15, 2010.

Prices and computation times (CPU time) as functions of the grid size are presented in table \ref{conv-nv-time} and figure \ref{fig:convergence} for four numerical schemes: second (II)  and third (III) order schemes using the grid described in the next paragraph and second (II') and third (III') order schemes using the optimal grids computed using formulas \eqref{eq:optimalgrid1} and \eqref{eq:optimalgrid2} with $p=3$, bootstrapped from the previous grid.

\newcommand{\withrelerr}[2]{$\displaystyle\mathop{#2}_{\tiny #1}$}
\begin{table}[!h]
\centering
\begin{tabular}{ccccc}
\toprule
Grid size & II & II' & III & III' \\
\midrule
50 \vspace{1mm}& \withrelerr{0.078}{4.713352}  &	 \withrelerr{0.156}{3.998833} & \withrelerr{0.094}{4.230729}  &  \withrelerr{0.219}{4.183711}  \\

100  \vspace{1mm}& \withrelerr{0.156}{4.614899}  &\withrelerr{0.438}{4.461027} &	\withrelerr{0.234}{4.492105}  &	 \withrelerr{0.641}{4.487217} \\

200  \vspace{1mm}& \withrelerr{0.531}{4.61589}   &	\withrelerr{1.641}{4.549288}  &	\withrelerr{0.781}{4.572986}  &	 \withrelerr{2.313}{4.548754} \\

300	\vspace{1mm} & \withrelerr{1.188}{4.594584}   &	\withrelerr{3.641}{4.578181}  &	\withrelerr{1.688}{4.572749}  &	 \withrelerr{4.969}{4.57237}\\

400	\vspace{1mm} & \withrelerr{2.063}{4.586547}   &	\withrelerr{6.203}{4.575781}  &	\withrelerr{2.844}{4.572500}  &	 \withrelerr{8.469}{4.572397} \\

500 \vspace{1mm}& \withrelerr{3.250}{4.582519}  &	\withrelerr{9.766}{4.574238}  &	\withrelerr{4.438}{4.572446}  &	 \withrelerr{13.734}{4.572311} \\

750 \vspace{1mm}& \withrelerr{6.859}{4.578084} &	\withrelerr{23.000}{4.573407} &	\withrelerr{10.469}{4.572371}  &	 \withrelerr{29.063}{4.572382} \\

1000 \vspace{1mm}& \withrelerr{12.234}{4.576436} &	\withrelerr{40.531}{4.572998} &	\withrelerr{17.813}{4.572351} &	 \withrelerr{50.844}{4.572309} \\

2000 \vspace{1mm}& \withrelerr{47.703}{4.574094} &	\withrelerr{166.281}{4.57253} &	\withrelerr{71.375}{4.572314} &	 \withrelerr{196.391}{4.572327} \\

3000 \vspace{1mm}& \withrelerr{112.016}{4.573509} &	\withrelerr{356.875}{4.572426} &	\withrelerr{163.547}{4.572333} &	 \withrelerr{436.609}{4.572319} \\
\bottomrule
\end{tabular}
\caption{\emph{Price of a put on a CPPI portfolio (top numbers) and CPU time in seconds (bottom) as a function of the grid size in four implementations of the algorithm: (II)~second order scheme,  (II')~second order scheme with computation of an optimal grid, (III)~third order scheme, (III')~third order scheme with computation of an optimal grid. }}
\label{conv-nv-time}
\end{table}

\begin{figure}[p]
\centering
\includegraphics[width=\textwidth]{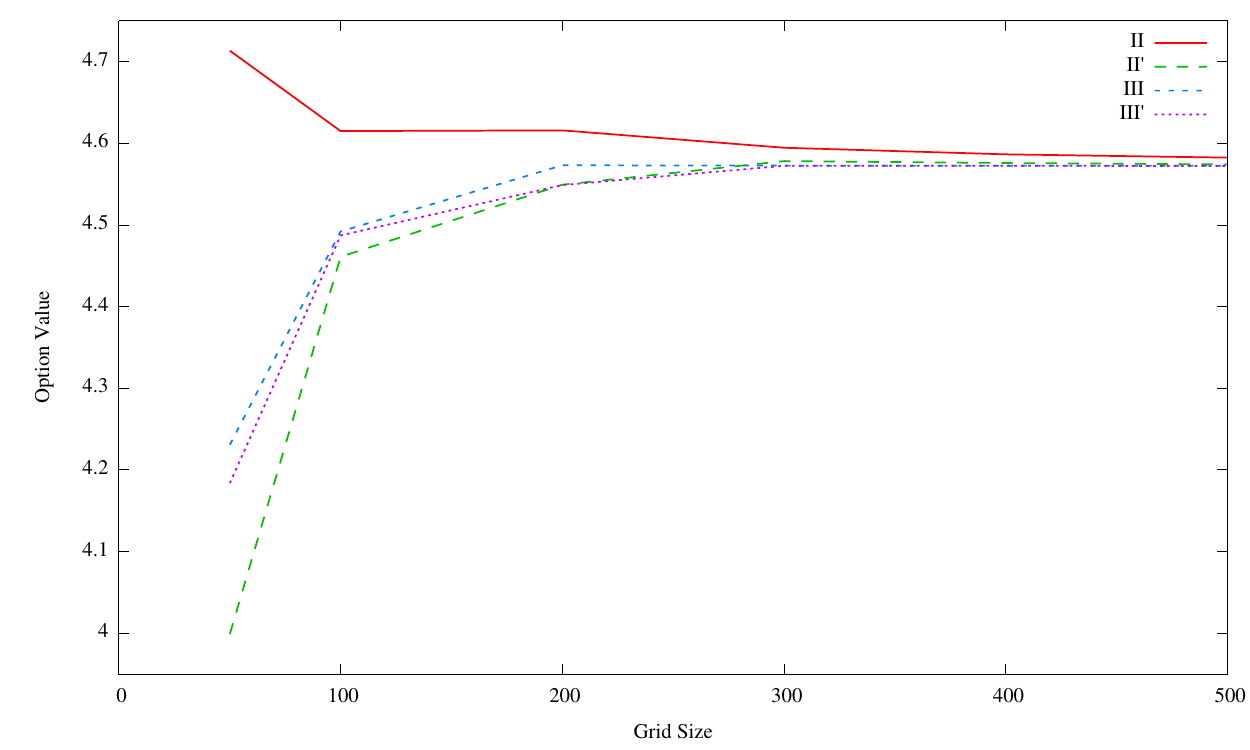}
\includegraphics[width=\textwidth]{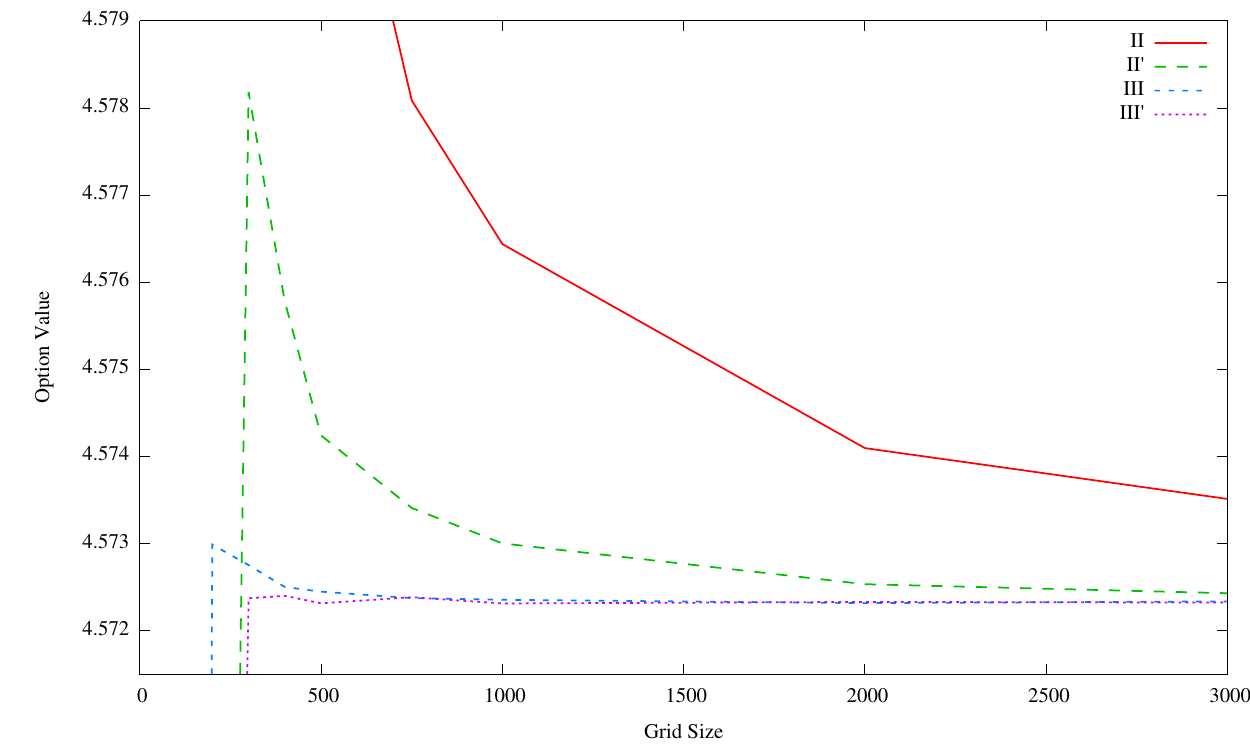}
\caption{\emph{Price of a put on a CPPI portfolio as a function of the grid size in four implementations of the algorithm: (II)~second order scheme,  (II')~second order scheme with computation of an optimal grid, (III)~third order scheme, (III')~third order scheme with computation of an optimal grid.}}
\label{fig:convergence}
\end{figure}

The grid used in cases (II) and (III) is built in the following way. We first evaluate an upper bound inverting the cumulative distribution  of a self-financing CPPI with natural bond floor and continuous reallocation: this distribution is a shifted lognormal distribution with volatility $m \sigma$. We set the bound such that the cumulative at this bound is $1-10^{-20}$. In the worst case where the risky asset drops to zero, the loss is equal to $m$ times the current portfolio value. Hence we set the lower bound to $(1-m)$ times the upper bound. These bounds are usually very large but the grid must be refined where there is convexity in the CPPI price, according to the error analysis of the previous section. We therefore take a parametric form for the grid between the bounds. We split the grid in three parts: a log-linear ($y=ae^{bx}$) central region which includes the bond floor value at all times and the strike of the option or the guarantee and two log-quadratic ($y=c + de^{fx^2}$) outer regions. Regions are glued together so that the parametric function describing the grid is $C^1$. This particular choice of the parametric form allows us to have sufficient refinement in the central region while reaching extreme bounds.

We observe that for all schemes, prices smoothly converge as the grid size is increased. As expected, third order schemes perform better than second order ones. Asymptotically, the computation of an optimal grid gives better results. However at lower grid sizes the error is larger; in addition the computation is almost three times longer. In practice, we therefore prefer scheme (III) which gives a smooth and fast convergence.

A precise price can be obtained\footnote{Computations were performed on a standard PC equipped with an Intel Core 2 6600 at 2.4~GHz and 3.50~GB of memory. We used the optimized Intel's Math Kernel Library to perform matrix-vector multiplications. Quoted times are CPU times.} in less than 1s. When the rates or the volatility are not constant, one needs to compute a transition matrix for each rebalancing period. In our test case, this means we generate 120 matrices. When all parameters are constant or piecewise constant, the number of matrices to be computed can be greatly reduced and the pricing is faster.

\section{Profit lock-in}
\label{sec:lockin}

\subsection{Pricing with profit lock-in}

CPPI strategies often include a profit lock-in feature: a fraction $\lambda$ of the performance is periodically locked in. $G$ therefore becomes a piecewise constant stochastic process $G(t)$. The floor is scaled accordingly as
\begin{equation*}
  H(t_i) = e^{-r(t_n-t_i)} G(t_i)
\end{equation*}
such that the terminal value of the CPPI portfolio should ideally be at least $G(t_n)$.

We denote by $t_{\ell_I}$, $I=1 \ldots N$ the subset of rebalancing dates at which the profit lock-in occurs (typically annually). For convenience, we set $\ell_0 = 0$ and $\ell_N = n$. Whenever the performance of the portfolio is positive between $t_{\ell_I}$ and $t_{\ell_{I+1}}$ the guarantee is raised to
\begin{equation}
\label{eq:ratchet}
  G\!\left( t_{\ell_{I+1}} \right) =  G\!\left( t_{\ell_I} \right) + \lambda \left[  C\!\left( t_{\ell_{I+1}} \right) -  C\!\left( t_{\ell_I} \right) \right]^+
\end{equation}
This expression can be recast into the more general form
\begin{equation}
\frac{G(t_{\ell_{I+1}})}{G(t_{\ell_I})} = f\!\left(\frac{C(t_{\ell_I})}{G(t_{\ell_I})}, \frac{C(t_{\ell_{I+1}})}{G(t_{\ell_I})} \right)
\label{eq:lockinF}
\end{equation}
with $f(x,y)$ a positive function.

The rebalancing rule
\begin{equation*}
  C^\text{risky}(t_i) = \max\!\left(m \frac{C(t_i) - e^{-r(t_n-t_i)} G(t_i)}{C(t_i)} \ , \ 0\right) C(t_i)
\end{equation*}
is of the more general form
\begin{equation*}
	C^\text{risky}(t_i) =  w\!\left(t_i,\frac{C(t_i)}{G(t_i)}\right) \, C(t_i)
\end{equation*}

We are concerned with the pricing of a European derivative product $V$ written on the CPPI portfolio whose terminal payoff $P$ might depend on $C$ and $G$ at maturity. The fair price at time $t_{\ell_i}$ can be written as the risk-neutral expectation
\begin{equation}
\label{eq:pricing2d}
V\!\big(t_{\ell_I}, C(t_{\ell_I}), G(t_{\ell_I}) \big) = e^{-r(t_{\ell_{I+1}}-t_{\ell_I})} \, \mathbb E \left[ V\!\big(t_{\ell_{I+1}}, C(t_{\ell_{I+1}}), G(t_{\ell_{I+1}}) \big) \, \big\vert \, \mathcal{F}_{t_{\ell_I}} \right]
\end{equation}
This problem is two-dimensional as one must know the joint law of $C$ and $G$ to apply this formula.
However, beautiful simplifications occur in the case of a homogeneous payoff of the form
\begin{equation}
\label{eq:payoff}
  P(C,G) = G \, \widetilde P\!\left(\frac{C}{G}\right)
\end{equation}
This includes the CPPI portfolio itself or options with strike in percentage of the final guarantee. Owing to the special form of the payoff we can state the following proposition.
\begin{prop}
\label{prop:lockin}
Consider the same hypothesis as in proposition \ref{prop:markov}, except for the fact that the guarantee is not constant but is rescaled at dates $t_{\ell_I}$ along equation \eqref{eq:lockinF}. For a European derivative product whose payoff at maturity is given by equation~\eqref{eq:payoff}, the fair price at date $t_{\ell_I}$ is homogeneous:
\begin{equation}
\label{eq:proposition1}
    V\!\big(t_{\ell_I}, C(t_{\ell_I}), G(t_{\ell_I}) \big) = G(t_{\ell_I}) \, \widetilde V\!\left(t_{\ell_I}, \frac{C(t_{\ell_I})}{G(t_{\ell_I})}\right)
\end{equation}
and $\widetilde V$ satisfies the recursion formula
\begin{equation}
\label{eq:proposition2}
    \widetilde V\!\left(t_{\ell_I}, \frac{C(t_{\ell_I})}{G(t_{\ell_I})}\right) = e^{-r(t_{\ell_{I+1}}-t_{\ell_I})} \int \mathrm d y \ \widetilde{\Phi}_{t_{\ell_I},t_{\ell_{I+1}}}\!\left(
    \frac{C(t_{\ell_I})}{G(t_{\ell_I})},y\right) \ \widetilde V\!\left(t_{\ell_{I+1}}, y\right)
\end{equation}
with the kernel
\begin{equation}
\label{eq:proposition3}
    \widetilde{\Phi}_{t_{\ell_I},t_{\ell_{I+1}}}\!\left(x, y\right) = \int \mathrm d z \ \delta\!\left( y - \frac{z}{f(x,z)}\right) f(x,z) \, \Phi_{t_{\ell_I},t_{\ell_{I+1}}}\!(x, z)
\end{equation}
In this equation, $\Phi_{t_{\ell_I},t_{\ell_{I+1}}}\!(x, z)$ is the transition operator defined in proposition \ref{prop:markov}.
\end{prop}

\begin{proof}
At time $t_n$, equation~\eqref{eq:proposition1} is verified due to equation~\eqref{eq:payoff}. We assume now equation~\eqref{eq:proposition1} holds at time $t_{\ell_{I+1}}$.
We denote by $t_{\ell_{I+1}}^-$ a time just before the profit lock-in, \emph{i.e.} before the jump in $G(t)$, but after possible jumps in the risky asset $S(t)$ (this describes the state of the system before applying the lock-in). As there is no profit lock-in between $t_{\ell_I}$ and $t_{\ell_{I+1}}^-$, we can use proposition~\ref{prop:markov} with transition operator $\Phi_{t_{\ell_I},t_{\ell_{I+1}}^-}$ and write
\begin{equation}
V\!\big(t_{\ell_I}, C(t_{\ell_I}), G(t_{\ell_I}) \big) = e^{-r(t_{\ell_{I+1}}-t_{\ell_I})} \int \mathrm d y \, \Phi_{t_{\ell_I},t_{\ell_{I+1}}^-}\!\left(\frac{C(t_{\ell_I})}{G(t_{\ell_I})} , y \right) V\!\left(t_{\ell_{I+1}}^-, G(t_{\ell_I}) y, G(t_{\ell_I}) \right)
\label{eq:lockinRec}
\end{equation}

When the lock-in is performed, the portfolio value is not changed:
\begin{equation*}
V\!\left(t_{\ell_{I+1}}^-, G(t_{\ell_I}) y, G\!\left(t_{\ell_I}\right) \right)
=
V\!\left(t_{\ell_{I+1}} , G(t_{\ell_I}) y , G\!\left(t_{\ell_{I+1}}\right) \right)
\end{equation*}
This value can be rewritten further using the recursion hypothesis~\eqref{eq:proposition1} as
\begin{equation}
V\!\left(t_{\ell_{I+1}}^-, G(t_{\ell_I}) y, G(t_{\ell_I}) \right) = G\!\left(t_{\ell_{I+1}}\right) \widetilde V\!\left(t_{\ell_{I+1}}, \frac{G(t_{\ell_I}) y}{G\!\left(t_{\ell_{I+1}} \right)} \right)
\label{eq:lockinV}
\end{equation}

The profit lock-in is deterministic as a function of $G(t_{\ell_I})$, $C(t_{\ell_I})$ and $C(t_{\ell_{I+1}})$. It is given by equation \eqref{eq:lockinF} which reads here
\begin{equation}
\label{eq:GIplus1}
G\!\left(t_{\ell_{I+1}}\right) = G\!\left(t_{\ell_I}\right) f\!\left( \frac{C(t_{\ell_I})}{G(t_{\ell_I})} , y \right)
\end{equation}
Using equations~\eqref{eq:lockinV} and~\eqref{eq:GIplus1}, equation \eqref{eq:lockinRec} is rewritten as
\begin{multline*}
V\!\big(t_{\ell_I}, C(t_{\ell_I}), G(t_{\ell_I}) \big) = \\
e^{-r(t_{\ell_{I+1}}-t_{\ell_I})} G(t_{\ell_I}) \int \mathrm d y \, \Phi_{t_{\ell_I},t_{\ell_{I+1}}^-}\!\left( \frac{C(t_{\ell_I})}{G(t_{\ell_I})} , y \right) f\!\left( \frac{C(t_{\ell_I})}{G(t_{\ell_I})},y\right) \widetilde V\!\left(t_{\ell_{I+1}}, \frac{y}{f\!\left( \frac{C(t_{\ell_I})}{G(t_{\ell_I})},y\right)} \right)
\end{multline*}

Inserting
\begin{equation*}
\int \mathrm d z \ \delta\!\left( z - \frac{y}{f\!\left( \frac{C(t_{\ell_I})}{G(t_{\ell_I})},y\right)}\right) = 1
\end{equation*}
we find that $V\!\big(t_{\ell_I}, C(t_{\ell_I}), G(t_{\ell_I}) \big)$ is indeed of the form given by equation \eqref{eq:proposition1} with equations \eqref{eq:proposition2} and \eqref{eq:proposition3} satisfied.
\end{proof}

\begin{rem}
Proposition~\ref{prop:lockin} can be generalized to a homogeneous payoff of degree $\alpha$:
\begin{equation*}
P(C,G) = G^\alpha \, \widetilde P\!\left(\frac{C}{G}\right)
\end{equation*}
The only modifications are in equation~\eqref{eq:proposition3} where $G\!\left(t_{\ell_I}\right)$ is replaced with $G\!\left(t_{\ell_I}\right)^\alpha$ and in equation~\eqref{eq:proposition3} where $f(x,z)$ is replaced with $f(x,z)^\alpha$.
\end{rem}

\begin{rem}
Proposition~\ref{prop:lockin} can also be generalized to a Bermudan payoff with exercise allowed at lock-in or reallocation dates: as usual, the maximum of values with and without exercising has to be taken during the backward propagation.
\end{rem}

\subsection{Numerical implementation}

The numerical implementation of proposition \ref{prop:lockin} is straightforward, building on the basis of section \ref{sec:numericalBasic}.
\begin{itemize}
\item
A grid is built as in the basic case of section \ref{sec:numericalBasic}, such that $C(t_0)/G(t_0) = h_{j_0}$ for some index $j_0$.
\item
From the local transition matrices $M^{(i,i+1)}_{jk}$ of section 1, we compute the transition matrices between two lock-in dates as the matrix product
\begin{equation*}
  M^{(\ell_I,\ell_{I+1})} \ =  \ M^{(\ell_I,\ell_I+1)} \ M^{(\ell_I+1,\ell_I+2)} \ \ldots \ M^{(\ell_{I+1}-1,\ell_{I+1})}
\end{equation*}
\item
Using the lock-in function $f(x,y)$, the discrete analogue of the kernel $\widetilde{\Phi}_{t_{\ell_I},t_{\ell_{I+1}}}(x,y)$ of equation \eqref{eq:proposition3} is a matrix $\widetilde{M}^{(I,I+1)}_{jk}$. We compute first
\begin{eqnarray*}
\widetilde{Q}^{(I,I+1)}_{jk} &=& \int_{h_k}^{h_{k+1}}\!\mathrm d y \, \widetilde{\Phi}_{(t_{l_I},t_{l_{I+1}})}(h_j,y)
\\
&=& \int\mathrm d z \, \mathbbm 1_{\left[h_k;h_{k+1}\right[}\!\left(\frac{z}{f(h_j,z)}\right) \, f(h_j,z) \, \Phi_{(t_{l_I},t_{l_{I+1}})}(h_j,z)
\\
&\simeq&  \sum_l \mathbbm 1_{\left[h_k;h_{k+1}\right[}\!\left(\frac{h_l}{f(h_j,h_l)}\right) \, f(h_j,h_l) \, M^{(\ell_I,\ell_{I+1})}_{jl}
\\
\widetilde{Q}^{1\, (I,I+1)}_{jk} &=& \int_{h_k}^{h_{k+1}}\!\mathrm d y \, \widetilde{\Phi}_{(t_{l_I},t_{l_{I+1}})}(h_j,y) \, y
\\
&=& \int\mathrm d z \, \mathbbm 1_{\left[h_k;h_{k+1}\right[}\!\left(\frac{z}{f(h_j,z)}\right) \, \Phi_{(t_{l_I},t_{l_{I+1}})}(h_j,z) \, z
\\
&\simeq&  \sum_l \mathbbm 1_{\left[h_k;h_{k+1}\right[}\!\left(\frac{h_l}{f(h_j,h_l)}\right) \, M^{(\ell_I,\ell_{I+1})}_{jl} \, h_l
\end{eqnarray*}
We define the matrix $\widetilde{M}^{(I,I+1)}$ from $\widetilde{Q}^{(I,I+1)}$ and $\widetilde{Q}^{1\, (I,I+1)}$ exactly as in section \ref{sec:numericalBasic}:
\begin{equation*}
\widetilde M^{(I,I+1)}_{jk} = \widetilde M^{+ \, (I,I+1)}_{jk} + \widetilde M^{- \, (I,I+1)}_{jk}
\end{equation*}
with
\begin{eqnarray*}
\widetilde M^{+ \, (I,I+1)}_{jk} &=& \frac{h_{k+1} \widetilde Q^{(I,I+1)}_{jk} - \widetilde Q^{1\, (I,I+1)}_{jk}}{h_{k+1}-h_k}
\\
\widetilde M^{- \, (I,I+1)}_{jk} &=& \frac{\widetilde Q^{1\, (I,I+1)}_{j,k-1} - h_{k-1} \widetilde Q^{(I,I+1)}_{j,k-1}}{h_k-h_{k-1}}
\end{eqnarray*}
\item
The rescaled payoff $\widetilde V$ of proposition \ref{prop:lockin} is discretized as $\widetilde V^{(I)}_j = \widetilde V(t_{\ell_I},h_j)$. It is propagated backward from lock-in date to lock-in date through equation \eqref{eq:proposition2} which becomes
\begin{equation*}
\widetilde V^{(I)}_j = e^{-r (t_{\ell_{I+1}} - t_{\ell_I})} \sum_k \widetilde M^{(I,I+1)}_{jk} \widetilde V^{(I+1)}_k
\end{equation*}
with terminal condition
\begin{equation*}
\widetilde V^{(N)}_j = \widetilde P(h_j)
\end{equation*}
\item
The fair price of the derivative at time 0 is given by
\begin{equation*}
V(t_0,C(t_0),G(t_0)) = G(t_0) \widetilde V^{(0)}_{j_0}
\end{equation*}

\end{itemize}

\subsection{Probability transition}

Proposition \ref{prop:lockin} provides a pricing algorithm for payoff which are homogeneous in $C,G$. The distribution of the portfolio value can also be computed, which may be useful for valuation and risk management of options with strike given as a fixed amount or for computing the final probability density function of a CPPI portfolio.
\begin{prop}
The probability transition operator between two lock-in dates $t_{\ell_I}$ and $t_{\ell_J}$
\begin{equation*}
  \Psi_{t_{\ell_I},t_{\ell_J}}(x,y) = \mathbb E\!\left[ \delta\!\left(\frac{C(t_{\ell_J})}{G(t_{\ell_I})} - y\right) \ \Big\vert \ \frac{C(t_{\ell_I})}{G(t_{\ell_I})} = x \right]
\end{equation*}
can be computed from the transition operator without profit lock-in as
\begin{equation}
  \Psi_{t_{\ell_I},t_{\ell_J}}(x,y) = \int \mathrm d z \ \frac{1}{f(x,z)}\ \Phi_{t_{\ell_I},t_{\ell_{I+1}}}(x,z) \ \Psi_{t_{\ell_{I+1}},t_{\ell_J}}\!\left(\frac{z}{f(x,z)},\frac{y}{f(x,z)}\right)
  \label{eq:lockinAbsolute}
\end{equation}
\end{prop}

\begin{proof}
This result is a rewriting of
\begin{multline*}
\mathbb E\!\left[ \delta\!\left(\frac{C(t_{\ell_J})}{G(t_{\ell_I})} - y\right) \ \Big \vert \ \frac{C(t_{\ell_I})}{G(t_{\ell_I})} = x \right] =
\\
\int \mathrm d z \ \mathbb E\!\left[ \delta\!\left(\frac{C(t_{\ell_{I+1}})}{G(t_{\ell_I})} - z\right) \ \Big\vert \ \frac{C(t_{\ell_I})}{G(t_{\ell_I})} = x \right]
\mathbb E\!\left[ \delta\!\left(\frac{C(t_{\ell_J})}{G(t_{\ell_{I+1}})} \frac{G(t_{\ell_{I+1}})}{G(t_{\ell_I})} - y \right) \ \Big\vert \ \frac{C(t_{\ell_{I+1}})}{G(t_{\ell_{I+1}})} \frac{G(t_{\ell_{I+1}})}{G(t_{\ell_I})} = z \right]
\end{multline*}
using equation~\eqref{eq:lockinF}.
\end{proof}

Supposing for simplicity that $\ell_0=0$ and $\ell_N=n$, the distribution of the portfolio at final maturity $t_n$ can then be obtained as
\begin{equation}
  \mathbb E\Big[ \delta\big(C(t_n)- y\big) \ \big\vert \ C(t_0) = x \Big]
  =
  \frac{1}{G(t_0)} \Psi_{t_0,t_n}\!\left( \frac{x}{G(t_0)} ,  \frac{y}{G(t_0)}  \right)
\end{equation}

The main difference with the homogeneous payoff is the following: whereas the dynamics of $C(t_{\ell_I})/G(t_{\ell_I})$ taken at lock-in dates is Markovian, it is not true for the dynamics of $C(t_{\ell_I})$ itself.

The numerical implementation is similar to the other cases. Equation~\eqref{eq:lockinAbsolute} is not a product of operators. As a consequence, it is also numerically more time-consuming.

\section{Numerical examples}

In order to illustrate the smoothness of the results, we plot in figure \ref{distrib} the probability distribution at maturity for three different CPPI strategies. We take a risky asset which follows the jump-diffusion model introduced by Kou \cite{kou:jdm}. There are upward and downward jumps, both exponentially distributed with distinct parameters. The diffusion has volatility 20\%, downward jump with intensity 0.1~$\mathrm{yr}^{-1}$ and mean value 0.1, upward jumps with intensity 0.1~$\mathrm{yr}^{-1}$ and mean value 0.05. We consider strategies on this risky asset with a maturity of 10 years and weekly rebalancing. The risk-free rate gives an initial zero-coupon value of 0.606.

\begin{itemize}
\item
We consider first the vanilla case of section \ref{sec:basic} with multiplier $m=4$:
\[
    w \big(t, C \big) = \max\!\left( m \frac{C-H(t)}{C}, 0 \right)
\]
\item
We add a cap (150\%) to the risky asset exposure so that
\[
    w \big(t, C \big) = \min\!\left( \text{Cap}\, , \, \max\!\left( m \frac{C-H(t)}{C}, 0 \right) \right)
\]
\item
Removing the cap, we finally incorporate to the strategy an annual lock-in of 75\% as described in section \ref{sec:lockin}.
\end{itemize}

We refer the reader to \cite{paulot2009epc} for further discussions of the effects of the various features commonly found in CPPI-based contracts.

\begin{figure}[!h]
\centering
\includegraphics[width=\textwidth]{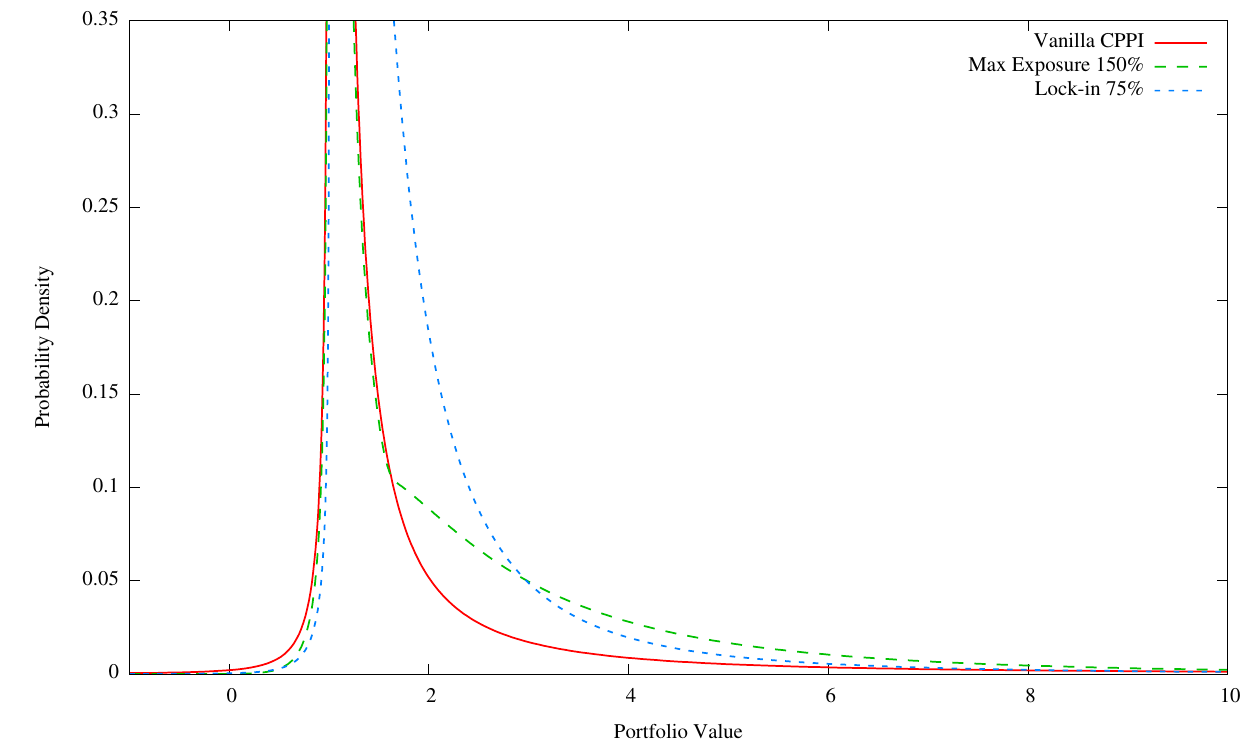}
\includegraphics[width=\textwidth]{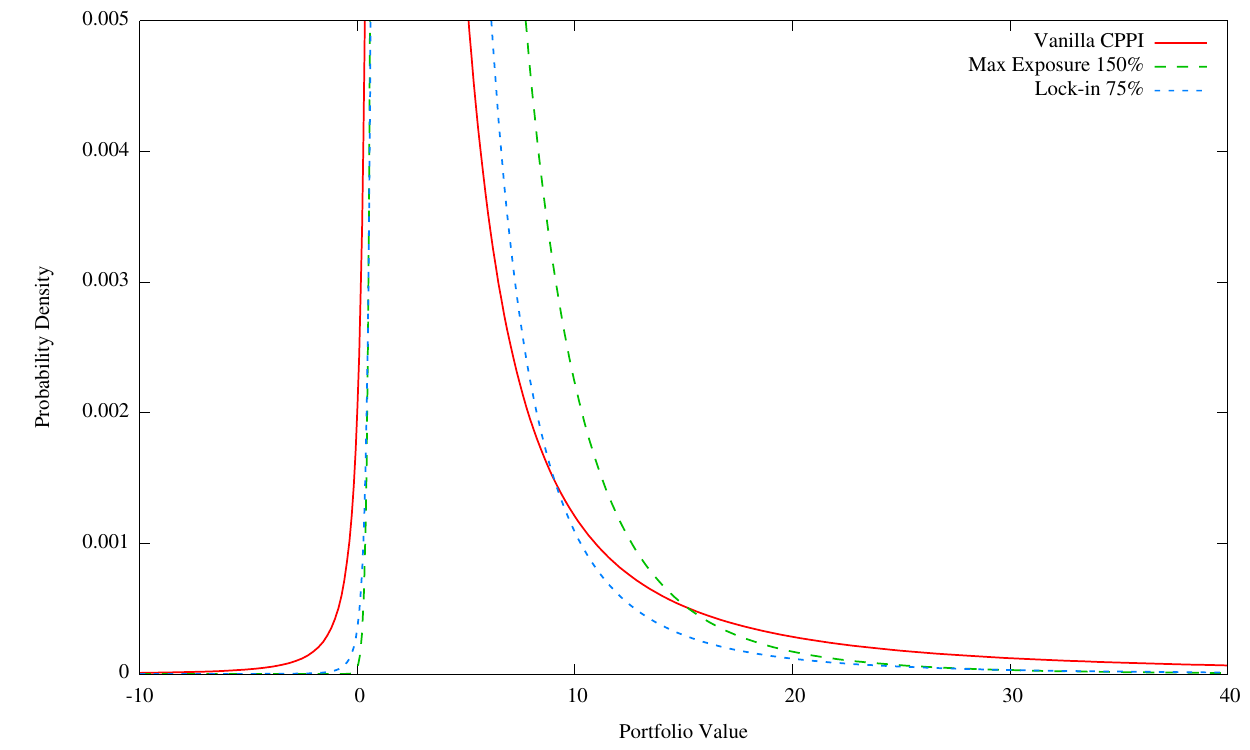}
\caption{\emph{Probability Density of a CPPI portfolio value at maturity for three different strategies. For the first curve, the CPPI has maturity 10 years, weekly rebalancing, multiplier 4. The initial investment and the initial guaranteed level are 1. The mean price of the strategy is 1.65. The risky asset follows a Brownian motion with 20\% constant volatility with additional jumps (Kou model): downward jump with intensity 0.1~$\mathrm{yr}^{-1}$ and mean value 0.1, upward jumps with intensity 0.1~$\mathrm{yr}^{-1}$ and mean value 0.05. For the second curve the risky asset weighting is capped to 150\%. For the third one there is no cap on the exposure but a profit lock-in of 75\% of the performance is performed annually. The second graph is a zoom on the tails of the distribution.}}
\label{distrib}
\end{figure}

\FloatBarrier

\section{Conclusion}

We introduced an efficient scheme for the pricing of CPPI strategies and options. Instead of following the underlying spot, the portfolio value, the risky asset weighting and the guaranteed amount as in a Monte-Carlo pricing or a classical PIDE backward propagation, only the CPPI portfolio value at rebalancing dates is considered (or the ratio of the CPPI over the guaranteed amount in the case of lock-in). We proved that under the hypothesis of  independent increments of the logarithm of the underlying, the portfolio value at reallocation dates (or lock-in dates) is a discrete time Markov process in  one single variable. We derived a natural pricing scheme which uses this property to price both European and Bermudan derivatives on the CPPI portfolio. Numerical experiments exhibit a very fast convergence. This can be used to estimate the gap risk of guaranteed CPPI products, price options on CPPI strategies and design hedging strategies.

\pagebreak
\bibliographystyle{alpha}
\bibliography{CPPI2}

\end{document}